%% file: main.tex
\pdfoutput=1
\documentclass[a4paper,UKenglish]{lipics-v2021}

\input{macros}

\hideLIPIcs

\title{A positional \texorpdfstring{$\boldsymbol{\Pi}^{0}_{3}$}{Π⁰₃}-complete objective}

\author{Antonio Casares}{University of Warsaw, Poland}{antoniocasares@mimuw.edu.pl}{https://orcid.org/0000-0002-6539-2020}{}
\author{Pierre Ohlmann}{CNRS, Laboratoire d'Informatique et des Systèmes, Marseille, France}{pierre.ohlmann@lis-lab.fr}{https://orcid.org/0000-0002-4685-5253}{}
\author{Pierre Vandenhove}{LaBRI, Université de Bordeaux, France}{pierre.vandenhove@umons.ac.be}{https://orcid.org/0000-0001-5834-1068}{}

\authorrunning{A. Casares, P. Ohlmann, and P. Vandenhove}

\Copyright{Antonio Casares, Pierre Ohlmann, and Pierre Vandenhove}



\funding{Antonio Casares is funded by the Polish National Science Centre (NCN) grant \emph{Polynomial finite state computation} (2022/46/A/ST6/00072).
Pierre Vandenhove was funded by ANR project \emph{G4S} (ANR-21-CE48-0010-01).}

\acknowledgements{We acknowledge the use of GPT-5.5 Pro, which assisted in identifying a technical issue in the proof of Lemma~\ref{lem:embedSmallGraphs} presented in the first arXiv version of the paper (v1). The authors have corrected the issue in the present version.}

\nolinenumbers

\usepackage{tikz}
\usetikzlibrary{calc,arrows.meta}
\tikzset{
  circ/.style={draw,circle,minimum height=7mm}
}

\definecolor{darkblue}{rgb}{0.0, 0.0, 0.55}
\newcommand{\blue}[1]{{\color{darkblue}#1}}
\definecolor{darkred}{rgb}{0.65, 0.0, 0.0}
\newcommand{\red}[1]{{\color{darkred}#1}}

\begin{document}

\maketitle

\begin{abstract}
We study zero-sum, turn-based games on graphs.
In this note, we show the existence of a game objective that is $\bpi{3}$-complete for the Borel hierarchy and that is \emph{positional}, i.e., for which positional strategies suffice for the first player to win over game graphs of arbitrary cardinality.
To the best of our knowledge, this is the first known such objective; all previously known positional objectives are in~$\bsigma{3}$.
The objective in question is a qualitative variant of the well-studied \emph{total-payoff objective}, which deals with the behaviour of the sequence of partial sums of weights seen along plays.
\end{abstract}

\section{Context and contribution} \label{sec:intro}

We consider zero-sum, turn-based, two-player, infinite-duration games played on (possibly infinite) graphs~\cite{NathBook}.
In this context, two players, Eve and Adam, take turns moving a token along the edges of an edge-labelled directed game graph.
Throughout, we assume that graphs do not have dead-ends: every vertex has at least one outgoing edge.
This interaction thus produces an infinite path whose edge labels form an infinite word $\word$ in $C^\omega$, where $C$ is a countable, possibly infinite alphabet of labels.
An \emph{objective} $W \subseteq C^\omega$ is specified in advance; Eve wins the game if the word $\word$ belongs to~$W$; otherwise, Adam wins.

A strategy is called \emph{positional} (or memoryless) if it chooses the next move as a function only of the current vertex, regardless of past moves.
An objective is called \emph{positional}\footnote{The literature sometimes uses ``half-positional'' for this notion, since there is a requirement on Eve's strategy complexity, but not on Adam's.} if, in every game graph (of arbitrary cardinality), Eve has a positional strategy winning from every vertex from which she has a winning strategy.

To the best of our knowledge, all previously known positional objectives belong to $\bsigma{3}$, the (existential) third level of the Borel hierarchy in the usual product topology\footnote{We recall that the open sets of this topology are those of the form $LC^\omega$, for $L\subseteq C^*$ a set of finite words.} over $C^\omega$.
For instance, the positionality of the $\omega$-regular objectives is well understood~\cite{CO26}, but they all lie in $\bdelta{3} = \bsigma{3} \cap \bpi{3}$ (as shown in~\cite{BL69Definability}).
There are additional examples stemming from characterisations of (subclasses of) objectives in $\bsigma{1}$, $\bpi{1}$, and $\bsigma{2}$: see, respectively, \cite{BFRV23}, \cite{CFH24}, and~\cite{OS26}.
The following natural $\bsigma{3}$-complete objective is also shown to be positional in~\cite{COV26PositionalityMemory}:
\[\InfOcc = \{\word \in \N^\omega \mid \exists c\in\N, |\word|_c = \infty\},\] where $|\word|_c$ denotes the number of occurrences of $c$ in $\word$.
The objective $\InfOcc$ is thus the set of words in which some number occurs infinitely often.
However, its complement, which is a $\bpi{3}$-complete objective, is not positional --- to see it, consider a game graph with a single vertex where Eve has to choose among infinitely many self-loops, each labelled with a different number $c\in\N$.
This leads to the following question: does there exist a positional objective that does not belong to $\bsigma{3}$?

We answer this question positively by showing that the following $\bpi{3}$-complete objective over $C = \Z$ is positional:
\[
\STI = \{w_0w_1\ldots \in \Z^\omega \mid \lim_{k\to \infty}\sum_{i=0}^{k-1} w_i = +\infty\}.\]

This objective is a qualitative variant of a \emph{total-payoff objective} (also called a \emph{total-reward objective}), where the goal is usually to maximise the $\limsup$ or $\liminf$ of the sequence of partial sums of weights.
Total-payoff objectives are positional over \emph{finite} arenas~\cite{GZ04}.
However, over infinite arenas, they are in general not positional.
Various classes of strategies are needed, depending on the variant considered ($\limsup$ or $\liminf$, and with a rational, $+\infty$, or $-\infty$ threshold) and the class of arenas~\cite{OS26,BIPTV24}.
Objective $\STI$ is a natural variant that turns out to have a remarkably low strategy complexity even over the most general class of arenas.
Our results fill a gap in the understanding of quantitative objectives~\cite{BIPTV24}.

Note that $\STI$ is prefix-independent (i.e., for all $\word\in\Z^*$ and $\word'\in\Z^\omega$, $\word'\in\STI$ if and only if $\word\word'\in\STI$).

\begin{theorem}\label{thm:main}
	The objective $\STI$ is $\bpi{3}$-complete and positional.
\end{theorem}

The rest of the note is devoted to the proof of Theorem~\ref{thm:main}.
We quickly show in Section~\ref{sec:completeness} that $\STI$ is $\bpi{3}$-complete; our main contribution, in Section~\ref{sec:positionality}, is a positionality proof based on constructing \emph{almost-universal graphs} for $\STI$, and applying~\cite[Theorem~3.2]{Ohlmann23}.

A natural open question is whether every level of the Borel hierarchy admits a complete objective that is positional.

\section{\texorpdfstring{$\boldsymbol{\Pi}^{0}_{3}$}{Π⁰₃}-completeness of \texorpdfstring{$\mathsf{SumToInfinity}$}{SumToInfinity}} \label{sec:completeness}
We refer to~\cite{kechris1995Descriptive} for definitions of the Borel hierarchy.

To show that $\STI$ is in $\bpi{3}$, observe that
\[
	\STI = \bigcap_{n\in\N} \bigcup_{j \in \N} \bigcap_{k \ge j} \{w_0w_1\ldots\in\Z^\omega \mid \sum_{i=0}^{k-1} w_i \ge n\},
\]
where the inner sets $\{w_0w_1\ldots\in\Z^\omega \mid \sum_{i=0}^{k-1} w_i \ge n\}$ are clopen.

To show that $\STI$ is $\bpi{3}$-hard, we reduce the following $\bpi{3}$-hard objective~\cite[Ex.~23.2]{kechris1995Descriptive} to it:
\[ \FinOcc = \{ \word \in \N^\omega \mid \forall c\in \N, |\word|_c \text{ is finite} \}. \]

This objective is the complement of the objective $\InfOcc$ discussed above.
We recall that for a reduction from $\FinOcc$ to $\STI$, we need to show a continuous mapping $f\colon \N^\omega \to \Z^\omega$ such that $f^{-1}(\STI) = \FinOcc$. Such a mapping is defined by
\[ f(c_0c_1\dots) = w_0w_1\dots, \;\; \text{ with } w_i = c_{i+1} - c_i. \]

The function $f$ is continuous because if $\word,\word' \in \N^\omega$ are two words with a common prefix of length~$n+1$, then $f(\word)$ and $f(\word')$ share a prefix of length~$n$.

Let us show that $f^{-1}(\STI) = \FinOcc$. Note that
$\sum_{i=0}^{k-1} (c_{i+1}-c_i) = c_k - c_0$.
If $c_0c_1\ldots \notin \FinOcc$, then there is $d\in \N$ such that $c_k = d$ for infinitely many $k$. Thus, $c_k-c_0$ takes the same value $d - c_0\in \Z$ for infinitely many $k$. Therefore, $\sum_{i=0}^{k-1} (c_{i+1}-c_i) \nrightarrow +\infty$.
Conversely, if $c_0c_1\ldots \in \FinOcc$, then, for all $c\in \N$, $c_k-c_0 >c$ for all sufficiently large $k$, so $\sum_{i=0}^{k-1} (c_{i+1}-c_i) \to +\infty$.

\section{Positionality of \texorpdfstring{$\mathsf{SumToInfinity}$}{SumToInfinity}} \label{sec:positionality}

As explained in Section~\ref{sec:intro}, we prove Theorem~\ref{thm:main} using the theory of \emph{universal graphs}. To keep this note short, we include only the crucial formal definitions of universal graphs; we refer the reader to~\cite{Ohlmann23} for additional context.

In what follows, the word \emph{graph} stands for a directed graph with edges labelled by elements of~$C$, with no dead-end: every vertex has an outgoing edge.
The set of vertices of a graph $G$ is written $V(G)$, and its set of edges is written $E(G) \subseteq V(G) \times C \times V(G)$.
An edge from vertex $v$ to vertex $v'$ labelled by $c$ is written $v \re{c} v'$.

\subparagraph*{Almost-universality.}
A \emph{tree} is a graph with a root vertex $t_0$ such that all vertices admit a unique finite directed path from $t_0$.
A \emph{graph morphism} from $G$ to $H$ is a map $f\colon V(G) \to V(H)$ such that for every edge $v \re{c} v'$ in $G$, $f(v) \re{c} f(v')$ is an edge in $H$.
We write $G \to H$ if there exists such a morphism.
A \emph{well-ordered graph} is a graph equipped with a well-order~$\le$ on its vertices.
An ordered graph is \emph{monotone} if $u \geq v \re{c} v' \geq u'$ implies $u \re{c} u'$.
A graph \emph{satisfies} an objective $W$ if the labelling of every infinite path of the graph belongs to $W$.
For a prefix-independent objective $W$ and a cardinal $\kappa$, a graph $U$ is said to be \emph{$(\kappa,W)$-almost-universal} if
\begin{itemize}
    \item $U$ satisfies $W$, and
    \item for all trees $T$ of size $<\kappa$ satisfying $W$, there is a vertex $v_0$ such that $T[v_0] \to U$,
\end{itemize}
where $T[v_0]$ denotes the subtree induced by $v_0$ and its descendants.

The proof of Theorem~\ref{thm:main} will rely on the following result.
\begin{lemma}[{\cite[Theorem~3.2 and Lemma~4.5]{Ohlmann23}}] \label{lem:almost-universal-positional}
	Let $W$ be a prefix-independent objective.
    If, for all cardinals $\kappa$, there exists a well-ordered monotone $(\kappa,W)$-almost-universal graph, then~$W$ is positional.
\end{lemma}

By the prefix-independence of $\STI$ observed above, to prove its positionality, it therefore suffices to build a $(\kappa, \STI)$-almost-universal graph $U$ for every cardinal $\kappa$.
In what follows, let $C = \Z$ and let $\kappa$ be an infinite cardinal.

\subparagraph*{Definition of $U$.}
We identify $\kappa$ with its initial ordinal (i.e., the smallest ordinal of cardinality~$\kappa$) and consider finite tuples of ordinals below $\kappa$, i.e., elements of $\bigcup_{n<\omega} \kappa^n$.
For such a tuple $u$, we let $|u|$ denote the length of $u$; for instance, $|(0,1)|=2$.
For $i\le |u|$, we write~$u_{<i}$ for the restriction of~$u$ to its first $i$ coordinates:
\[
    (u_0,\dots,u_n)_{<i} = (u_0,\dots,u_{i-1}).
\]
Recall the lexicographic ordering:
\[
    u >_\lex u' \iff u' \text{ is a proper prefix of $u$, or } \exists i < \min\{|u|,|u'|\}, [u_{<i} = u'_{<i} \text{ and } u_i > u'_i].
\]

Consider the graph $U$ defined over $V(U)= \bigcup_{n<\omega} \kappa^n$ by
\[
    E(U) = \{u \re{w} u' \mid w\in\Z,\ |u| + w \geq |u'|, \text{ and } [|u| + w = |u'| \implies u >_\lex u']\}.
\]
Intuitively, the length of the tuples in $U$ encodes an underapproximation of the sum of weights along a given path: an edge either tracks precisely the sum of weights (when $|u| + w = |u'|$), or it underestimates it (when $|u| + w > |u'|$).
In the former case, there is the additional requirement that the tuple must stricly decrease, in the sense that $u >_\lex u'$.
In the latter case, the ordinals in the tuple can have any value.
These rules prevent in particular the existence of cycles with sum of weights $0$ in~$U$; to go back to the same vertex, some strict increase in the sum of weights is necessary.
Note that $U$ does not have dead-ends because, for every vertex~$u$, there is an edge $u \re{1} u$.

The order on $U$ is then defined by
\[
    u > u' \iff |u|>|u'| \text{ or } [|u|=|u'| \text{ and } u >_\lex u'].
\]
We define the non-strict order $\leq$ in the usual way: $u \leq u'$ if and only if $u = u'$ or $u < u'$.
We raise the reader's attention to the fact that the order on $U$ does not coincide with the lexicographic order: for instance, $(0,0) > (1)$ in $U$.
Also, the order $\le_\lex$ is not a well-order on $V(U)$ (since, for instance, $(1) >_\lex (0,1) >_\lex (0,0,1) >_\lex \dots$ is an infinite descending sequence), but the order $\leq$ is a well-order on $V(U)$.

\begin{lemma} \label{lem:well-ordered-monotone}
	The graph $(U,\leq)$ is a well-ordered monotone graph.
\end{lemma}
\begin{proof}
	It is immediate that the order on $U$ is well-founded and total.
	Let us check that~$U$ is monotone. Let $u, v, u', v' \in V(U)$ and $w\in\Z$ such that $u \geq v \re{w} v' \geq u'$ in~$U$.
	Then, $|u|+w\geq |v| + w \geq |v'| \geq |u'|$.
	If one of these inequalities is strict, then $|u|+w >|u'|$.
	Otherwise, $|u|+w = |u'|$ and $u \geq_\lex v >_\lex v' \geq_\lex u'$.
	In both cases, we conclude that~$u\re{w} u'$.
\end{proof}

\begin{figure}[t]
	\centering
	\begin{tikzpicture}[every node/.style={font=\small,inner sep=1pt}]
		\draw (0,0) node[circ] (v0) {\red{$0$}};
		\draw (v0.west) node[left=2pt,font=\scriptsize] {\blue{$(1)$}};
		\draw (v0.north west) node[above left=1.4pt] {$v_0$};
		\draw ($(v0)+(-1,-1.4)$) node[circ] (v1) {\red{$2$}};
		\draw (v1.north west) node[above left=1.4pt] {$v_1$};
		\draw (v1.west) node[left=2pt,font=\scriptsize] {\blue{$(0,1,4)$}};
		\draw ($(v1)+(1.5,0)$) node[circ] (v11) {\red{$1$}};
		\draw (v11.north) node[above=2pt,font=\scriptsize] {\blue{$(0,1)$}};
		\draw ($(v11)+(1.5,0)$) node[circ] (v12) {\red{$2$}};
		\draw (v12.north) node[above=2pt,font=\scriptsize] {\blue{$(0,0,1)$}};
		\draw ($(v12)+(1.5,0)$) node[circ] (v13) {\red{$3$}};
		\draw (v13.north) node[above=2pt,font=\scriptsize] {\blue{$(0,0,0,1)$}};
		\draw ($(v13)+(1.5,0)$) node[circ,draw=none] (v14) {$\cdots$};
		\draw ($(v1)+(-1,-1.4)$) node[circ] (v2) {\red{$3$}};
		\draw (v2.west) node[left=2pt,font=\scriptsize] {\blue{$(0,1,3,6)$}};
		\draw ($(v2)+(1.5,0)$) node[circ] (v21) {\red{$2$}};
		\draw (v21.north) node[above=2pt,font=\scriptsize] {\blue{$(0,1,3)$}};
		\draw ($(v21)+(1.5,0)$) node[circ] (v22) {\red{$1$}};
		\draw (v22.north) node[above=2pt,font=\scriptsize] {\blue{$(0,1)$}};
		\draw ($(v22)+(1.5,0)$) node[circ] (v23) {\red{$2$}};
		\draw (v23.north) node[above=2pt,font=\scriptsize] {\blue{$(0,0,1)$}};
		\draw ($(v23)+(1.5,0)$) node[circ,draw=none] (v24) {$\cdots$};
		\draw ($(v2)+(-1,-1.4)$) node[circ] (v3) {\red{$4$}};
		\draw (v3.west) node[left=2pt,font=\scriptsize] {\blue{$(0,1,3,5,8)$}};
		\draw ($(v3)+(1.5,0)$) node[circ] (v31) {\red{$3$}};
		\draw (v31.north) node[above=2pt,font=\scriptsize] {\blue{$(0,1,3,5)$}};
		\draw ($(v31)+(1.5,0)$) node[circ] (v32) {\red{$2$}};
		\draw (v32.north) node[above=2pt,font=\scriptsize] {\blue{$(0,1,3)$}};
		\draw ($(v32)+(1.5,0)$) node[circ] (v33) {\red{$1$}};
		\draw (v33.north) node[above=2pt,font=\scriptsize] {\blue{$(0,1)$}};
		\draw ($(v33)+(1.5,0)$) node[circ] (v34) {\red{$2$}};
		\draw (v34.north) node[above=2pt,font=\scriptsize] {\blue{$(0,0,1)$}};
		\draw ($(v34)+(1.5,0)$) node[circ,draw=none] (v35) {$\cdots$};
		\draw ($(v3)+(-1,-1.4)$) node[circ] (v4) {\red{$5$}};
		\draw (v4.west) node[left=2pt,font=\scriptsize] {\blue{$(0,1,3,5,7,10)$}};
		\draw ($(v4)+(1.5,0)$) node[circ] (v41) {\red{$4$}};
		\draw (v41.south) node[below=2pt,font=\scriptsize] {\blue{$(0,1,3,5,7)$}};
		\draw ($(v41)+(1.5,0)$) node[circ] (v42) {\red{$3$}};
		\draw (v42.south) node[below=2pt,font=\scriptsize] {\blue{$(0,1,3,5)$}};
		\draw ($(v42)+(1.5,0)$) node[circ] (v43) {\red{$2$}};
		\draw (v43.south) node[below=2pt,font=\scriptsize] {\blue{$(0,1,3)$}};
		\draw ($(v43)+(1.5,0)$) node[circ] (v44) {\red{$1$}};
		\draw (v44.south) node[below=2pt,font=\scriptsize] {\blue{$(0,1)$}};
		\draw ($(v44)+(1.5,0)$) node[circ] (v45) {\red{$2$}};
		\draw (v45.south) node[below=2pt,font=\scriptsize] {\blue{$(0,0,1)$}};
		\draw ($(v45)+(1.5,0)$) node[circ,draw=none] (v46) {$\cdots$};
		\draw ($(v4)+(-1,-1.4)$) node[circ,draw=none,rotate=54.46] (v5) {$\dots$};

		\draw (v0) edge[-latex] node[above left] {$2$} (v1);
		\draw (v1) edge[-latex] node[above left] {$1$} (v2);
		\draw (v2) edge[-latex] node[above left] {$1$} (v3);
		\draw (v3) edge[-latex] node[above left] {$1$} (v4);
		\draw (v4) edge[-latex] node[above left] {$1$} (v5);

		\draw (v1) edge[-latex] node[above=2pt] {$-1$} (v11);
		\draw (v11) edge[-latex] node[above=2pt] {$1$} (v12);
		\draw (v12) edge[-latex] node[above=2pt] {$1$} (v13);
		\draw (v13) edge[-latex] node[above=2pt] {$1$} (v14);

		\draw (v2) edge[-latex] node[above=2pt] {$-1$} (v21);
		\draw (v21) edge[-latex] node[above=2pt] {$-1$} (v22);
		\draw (v22) edge[-latex] node[above=2pt] {$1$} (v23);
		\draw (v23) edge[-latex] node[above=2pt] {$1$} (v24);

		\draw (v3) edge[-latex] node[above=2pt] {$-1$} (v31);
		\draw (v31) edge[-latex] node[above=2pt] {$-1$} (v32);
		\draw (v32) edge[-latex] node[above=2pt] {$-1$} (v33);
		\draw (v33) edge[-latex] node[above=2pt] {$1$} (v34);
		\draw (v34) edge[-latex] node[above=2pt] {$1$} (v35);

		\draw (v4) edge[-latex] node[above=2pt] {$-1$} (v41);
		\draw (v41) edge[-latex] node[above=2pt] {$-1$} (v42);
		\draw (v42) edge[-latex] node[above=2pt] {$-1$} (v43);
		\draw (v43) edge[-latex] node[above=2pt] {$-1$} (v44);
		\draw (v44) edge[-latex] node[above=2pt] {$1$} (v45);
		\draw (v45) edge[-latex] node[above=2pt] {$1$} (v46);
	\end{tikzpicture}
	\caption{Tree $T$ used in Example~\ref{ex:morphismExample}.
		The edge weights are shown in black.
		This tree satisfies $\STI$, as every infinite path has a suffix labelled $1^\omega$.
		The value in \red{red} inside each vertex is the value $s(\cdot)$ defined in the proof of Lemma~\ref{lem:embedSmallGraphs}.
		The top vertex~$v_0$ is such that all nonempty finite paths from it have positive weight, so we assume it is the vertex given by Claim~\ref{claim:exists-base-vtx}.
		The tuples in \blue{blue} next to vertices correspond to the morphism to $U$ (for any fixed infinite cardinal $\kappa$) built in the proof of Lemma~\ref{lem:embedSmallGraphs}.}
	\label{fig:morphismExample}
\end{figure}

\begin{example} \label{ex:morphismExample}
	Before proving the $(\kappa, \STI)$-almost-universality of $U$, we give an example of a morphism of a tree into $U$.
	We consider the tree $T$ from Figure~\ref{fig:morphismExample}.
	The \blue{blue} tuples next to each vertex indicate a morphism from $T$ to $U$.
	The morphism given is exactly the one built by our proof of almost-universality below; we refer back to this example during the proof.
	\lipicsEnd
\end{example}

\subparagraph*{Almost-universality of $U$.}
We now prove the following.

\begin{theorem} \label{thm:almost-universal}
The graph $U$ is $(\kappa, \STI)$-almost-universal.
\end{theorem}

We prove the two conditions for almost-universality in two separate lemmas,
Lemmas~\ref{lem:satisfyW} and~\ref{lem:embedSmallGraphs}.

\begin{lemma}\label{lem:satisfyW}
	The graph $U$ satisfies $\STI$.
\end{lemma}
\begin{proof}
	Take an infinite path $u^0 \re{w_0} u^1 \re{w_1} \dots$ in $U$.
	For all $i$, let
	\[
	b_i = \begin{cases}
		0 &\text{if } |u^i|+ w_i= |u^{i+1}|, \\
		1 &\text{otherwise.}
	\end{cases}
	\]
	For all $i$, we have $|u^i|+w_i \geq |u^{i+1}|+b_i$.
	Therefore, for all $k$,
	\[
	\sum_{i=0}^{k-1}w_i \geq |u^{k}| - |u^{0}| + \sum_{i=0}^{k-1} b_i \geq \sum_{i=0}^{k-1} b_i - |u^{0}|.
	\]
	If $\sum_{i=0}^{k-1} b_i$ goes to $+\infty$ then, by the above, it also holds that $\sum_{i=0}^{k-1} w_i$ goes to $+\infty$, as wanted.

	So we assume otherwise: there is $i_0$ such that for all $i \geq i_0$, $b_i=0$.
	Then, for all $i \geq i_0$, we have $|u^{i}| + w_i = |u^{i+1}|$ and thus $u^i >_\lex u^{i+1}$.
	Consider the sequence $(|u^i|)_{i \geq i_0}$.
	We show that it converges to $+\infty$.
	Assume for contradiction that it does not: then, there is a sequence $j_0 < j_1 < j_2 < \dots$, with $j_0 \geq i_0$, such that the subsequence $(|u^{j_k}|)_{k \in \N}$ is constant, say with value~$n$.
	By transitivity of $>_\lex$, we have $u^{j_0} >_\lex u^{j_1} >_\lex \dots$.
	However, the set of tuples of length~$n$ is well-ordered by $\le_\lex$, so this is a contradiction.
	We conclude that $|u^i| \to +\infty$.
	Thus, for $k \ge i_0$, $\sum_{i=i_0}^{k-1} w_i = |u^{k}| - |u^{i_0}|$ converges to $+\infty$.
	This extends to $\sum_{i=0}^{k-1} w_i$ converging to~$+\infty$, as wanted.
\end{proof}

\begin{lemma}\label{lem:embedSmallGraphs}
    For all trees $T$ of cardinality $< \kappa$ satisfying $\STI$, there exists a vertex $v_0$ of $T$ such that $T[v_0] \to U$.
\end{lemma}

\begin{proof}
Let $T$ be a tree of cardinality $< \kappa$ satisfying $\STI$.
Given a finite path $\pi$, its \emph{weight}, denoted $w(\pi)$, is the sum of the edge weights appearing along $\pi$.
We start with the following claim.

\begin{claim}\label{claim:exists-base-vtx}
	There exists a vertex from which every nonempty finite path has weight $>0$.
\end{claim}

\begin{claimproof}
	Suppose for contradiction that this is not the case.
	Then starting from the root $t_0$ of~$T$, we get a proper descendant $t_1$ of $t_0$ such that the path from $t_0$ to $t_1$ has weight~$\leq 0$.
	Recursively, after choosing $t_n$, choose a proper descendant $t_{n+1}$ of $t_n$ such that the path from $t_n$ to $t_{n+1}$ has weight $\leq 0$.
	Iterating, we build an infinite branch of $T$ such that at each $t_n$, the accumulated weight of the path from $t_0$ is $\leq 0$.
	This contradicts that $T$ satisfies $\STI$.
\end{claimproof}

Fix a vertex $v_0$ such that, as given by Claim~\ref{claim:exists-base-vtx}, every nonempty finite path from $v_0$ has positive weight.
We construct a morphism $\phi$ from $T[v_0]$ to $U$.
Define $s(v_0)=0$ and, for every proper descendant $v$ of $v_0$, let $s(v)>0$ be the weight of the path from $v_0$ to $v$.

For every $v \in T[v_0]$, we set the length of $\phi(v)$ to be $s(v)+1$, i.e., we have $\phi(v)=(\phi_0(v),\dots,\phi_{s(v)}(v))$, where the ordinal coordinates $\phi_k(v)$ will be defined below.
Note that this entails that for every edge $v \re{w} v'$ in $T[v_0]$, we have $|\phi(v)| + w = |\phi(v')|$.
We still have to choose the ordinal coordinates $\phi_k(v)$ so that $\phi(v) >_{\lex} \phi(v')$.

To this end, we proceed as follows.
Let $v \in T[v_0]$ and let $k \leq s(v)$.
Intuitively, the ordinal~$\phi_k(v)$ will record ``how many times'' (as an ordinal), at most, vertices $v'$ such that $s(v') \leq k$ can be visited along a path from~$v$.

Formally, consider the forest $F_{k,v}$ (a \emph{forest} is a disjoint union of trees)
\begin{itemize}
	\item the vertices of which are the descendants $x$ of $v$, possibly including $v$ itself, such that $s(x) \leq k$, and
	\item where there is an edge $x \to y$ in $F_{k,v}$ whenever $y$ is a proper descendant of $x$ such that the unique path from $x$ to $y$ contains only internal vertices $z$ such that $s(z)>k$.
	Possibly, the path from $x$ to $y$ consists of a single edge, in which case it does not contain any internal vertex.
\end{itemize}
Observe that $v$ is a vertex of~$F_{k,v}$ if and only if $k = s(v)$ (in which case $F_{k,v}$ contains a single tree, with root $v$).

The forest $F_{k,v}$ is well-founded (i.e., it does not have infinite branches).
Indeed, an infinite path in $F_{k,v}$ would expand to an infinite branch of $T[v_0]$ containing infinitely many vertices $x$ such that $s(x)\le k$.
However, this contradicts that $T[v_0]$ satisfies $\STI$, because $s(x)$ is the accumulated weight from $v_0$ and tends to $+\infty$.

Define the rank $r(x)$ of a vertex in a well-founded forest to be $0$ for leaves and the ordinal
\[
	\sup\{r(y)+1 \mid x \to y\}
\]
for non-leaves.
The rank $r(F)$ of a forest $F$ is then
\[
	\sup\{r(x)+1 \mid x \text{ a vertex of } F\},
\] with the convention that $\sup \emptyset = 0$ (so the empty forest has rank $0$). Equivalently, this is the rank of an artificial root added to the forest, with edges to the roots of all trees in the forest.
Then, let
\[
	\phi_k(v) = r(F_{k,v}).
\]

We check that the map $\phi$ indeed takes values in $V(U) = \bigcup_{n<\omega} \kappa^n$.
Equivalently, we check that $\phi_k(v) < \kappa$ for all $v$ and $k$.
Since $T$ has cardinality $<\kappa$, the forest $F_{k,v}$ has cardinality $<\kappa$ and thus rank $<\kappa$.

\begin{example}
The construction of $\phi$ is illustrated in Example~\ref{ex:morphismExample} and Figure~\ref{fig:morphismExample}.
We explain the construction of $\phi$ for the vertex~$v_1$ such that $\phi(v_1) = (0, 1, 4)$.
We have $s(v_1) = 2$, so $\phi(v_1)$ has length~$3$ and $\phi_k(v_1)$ must be defined for $k = 0, 1, 2$.
We have $\phi_0(v_1) = r(F_{0,v_1}) = 0$ since $F_{0,v_1}$ is empty: there is no descendant $x$ of $v_1$ with $s(x) \le 0$.
We have $\phi_1(v_1) = r(F_{1,v_1}) = 1$ since $F_{1,v_1}$ contains infinitely many isolated vertices, namely the descendants $x$ of $v_1$ with $s(x) = 1$.
Each such vertex is a leaf of rank~$0$, so the rank of $F_{1,v_1}$ is~$1$.
Finally, we have $\phi_2(v_1) = r(F_{2,v_1}) = 4$. The forest $F_{2,v_1}$ is a tree with root~$v_1$. On each horizontal branch of~$T$, the vertices $x$ such that $s(x) \le 2$ form a three-vertex path with the pattern $\red{2} \re{-1} \red{1} \re{1} \red{2}$ in~$T$. Hence, the root $v_1$ has rank~$3$ and the rank of $F_{2,v_1}$ is~$4$.\lipicsEnd
\end{example}

\begin{claim}
	The map $\phi\colon V(T[v_0]) \to V(U)$ defines a morphism from $T[v_0]$ to $U$.
\end{claim}

\begin{claimproof}
	Consider an edge $v \re{w} v'$ in $T[v_0]$.
	We show that $\phi(v) \re{w} \phi(v')$ is an edge in $U$.
	Since $|\phi(v)|+w=|\phi(v')|$, we must prove that $\phi(v)>_{\lex} \phi(v')$.

	We will use the following property: for every $k < s(v)$ such that $k \leq s(v')$, observe that $F_{k,v'}$ is a union of connected components of $F_{k,v}$, since $s(v)>k$.
	Therefore, $\phi_k(v) = r(F_{k,v}) \ge r(F_{k,v'}) = \phi_k(v')$ in that case.

	We distinguish two cases.
	\begin{itemize}
		\item If $s(v)>s(v')$ then, by the above observation, $\phi_k(v)\ge\phi_k(v')$ for all $0\le k\le s(v')$.
		Then, either there is $k \le s(v')$ such that $\phi_k(v)>\phi_k(v')$, or $\phi_k(v)=\phi_k(v')$ for all $0\le k\le s(v')$, in which case $\phi(v')$ is a proper prefix of $\phi(v)$.
		In either case, $\phi(v)>_{\lex} \phi(v')$.
		\item Otherwise, $s(v) \leq s(v')$.
		Consider $k = s(v)$.
		The forest $F_{k,v}$ contains $v$ as a root.
		The roots (or the single root if $k = s(v')$) of $F_{k,v'}$ are children of $v$ in $F_{k,v}$.
		Hence, the rank of $v$ in $F_{k,v}$ is strictly larger than the rank of any root of $F_{k,v'}$.
		It follows that $\phi_k(v) = r(F_{k,v}) > r(F_{k,v'}) = \phi_k(v')$.

		Combining it with the fact that $\phi_k(v)\ge \phi_k(v')$ for all $0\le k< s(v)$, we deduce that $\phi(v)>_{\lex} \phi(v')$.\claimqedhere
	\end{itemize}
\end{claimproof}

This proves Lemma~\ref{lem:embedSmallGraphs}.
\end{proof}

Lemmas~\ref{lem:satisfyW} and~\ref{lem:embedSmallGraphs}
prove that $U$ is $(\kappa,\STI)$-almost-universal,
and hence prove Theorem~\ref{thm:almost-universal}.

Together with Lemma~\ref{lem:well-ordered-monotone}, this gives, for every cardinal $\kappa$,
a well-ordered monotone $(\kappa,\STI)$-almost-universal graph.
Lemma~\ref{lem:almost-universal-positional} therefore implies that $\STI$ is positional.
Section~\ref{sec:completeness} establishes its
$\bpi{3}$-completeness, completing the proof of
Theorem~\ref{thm:main}.

\bibliography{bib}

\end{document}

%% file: macros.tex
\newcommand{\re}[1]{\xrightarrow{#1}}

\newcommand{\N}{\mathbb N}
\newcommand{\Z}{\mathbb Z}

\newcommand{\STI}{\mathsf{SumToInfinity}}
\newcommand{\InfOcc}{\mathsf{InfOcc}}
\newcommand{\FinOcc}{\mathsf{FinOcc}}

\newcommand{\boldclass}[3]{\boldsymbol{#1}^{#2}_{#3}}

\newcommand{\bsigma}[1]{\boldclass{\Sigma}{0}{#1}}
\newcommand{\bpi}[1]{\boldclass{\Pi}{0}{#1}}
\newcommand{\bdelta}[1]{\boldclass{\Delta}{0}{#1}}

\newcommand{\lex}{\text{lex}}

\newcommand{\word}{x}